\documentclass[12pt]{article}
\usepackage{amsfonts,amsthm,amssymb}
\usepackage[fleqn]{amsmath}
\usepackage{hyperref}
\newcommand{\be}{\begin{eqnarray}}
\newcommand{\ee}{\end{eqnarray}}
\newcommand{\bez}{\begin{eqnarray*}}
\newcommand{\eez}{\end{eqnarray*}}

\newcommand{\cM}{\mathcal{M}}

\newcommand{\pa}{\partial}

\newcommand{\ins}{\leavevmode 
    \vbox{\kern.2em \hrule width1.2ex height0.1ex} 
    \hbox{\vrule width0.1ex height1.2ex depth0.ex \kern.1em} \,}

\theoremstyle{definition}
\newtheorem{theorem}{Theorem}[section]
\newtheorem{proposition}[theorem]{Proposition}

\newtheorem{remark}[theorem]{Remark}

\newtheorem{definition}[theorem]{Definition}
\newtheorem{lemma}[theorem]{Lemma}

\numberwithin{theorem}{section}
\numberwithin{equation}{section}

\allowdisplaybreaks

\begin{document}

\title{\textbf{On the Lorenzoni-Magri hierarchy \\ of hydrodynamic type }}

\author{\sc{Folkert M\"uller-Hoissen} \\
         \small  Institute for Theoretical Physics,  University of G\"ottingen \\
        \small  Friedrich-Hund-Platz 1, 37077 G\"ottingen, Germany \\
         \small  folkert.mueller-hoissen@phys.uni-goettingen.de  } 

\date{ }

\maketitle

\begin{abstract}
In 2005 Lorenzoni and Magri showed that a hydrodynamic-type hierarchy determined by the powers of a type $(1,1)$ 
tensor field (on a smooth manifold) with vanishing Nijenhuis torsion can be deformed to a more general hierarchy, 
with the help of a chain of conservation laws of the new hierarchy. We review this construction. The $(1,1)$ tensor fields 
of the resulting hierarchy have non-vanishing Nijenhuis torsion, in general, but their Haantjes tensor vanishes. 
\end{abstract}

\section{Introduction}
Systems of hydrodynamic type in two dimensions are first order homogeneous quasilinear partial differential equations (PDEs)
of the form
\bez
           \frac{ \pa x^\mu}{ \pa t} + N^\mu{}_\nu \, \frac{ \pa x^\nu}{ \pa y} = 0 \qquad  \mu=1,\ldots,n \, .
\eez
They play a role in various areas of mathematics and physics. Cases, in which they are 
integrable in some sense, are of particular interest.
The  seminal work of Tsarev \cite{Tsar91} (also see \cite{Pavl07} and references cited there) 
addressed the integrability of such systems essentially only in the case where the matrix $(N^\mu{}_\nu)$ of 
functions is diagonalizable.  Meanwhile there is quite a number of examples of non-diagonalizable hydrodynamic-type 
systems with integrability properties (see \cite{Fera93PD,Fera94TMP,Mokh+Fera96,Koda+Kono16,Pavl19,Xu+Fera20,Verg+Fera24}, for example).  

 One of the characteristic features of integrable PDEs is that they belong 
to an infinite hierarchy of compatible equations, such that their flows commute, allowing for common solutions. 

 The above system is invariant under general coordinate transformations if $N^\mu{}_\nu $ are the components of a 
tensor field of type $(1,1)$ on an $n$-dimensional manifold $\cM$ with local coordinates $x^\mu$. Such tensor fields 
are the central objects of Fr\"olicher-Nijenhuis theory  \cite{Froe+Nije56}. If the Nijenhuis torsion $T(N)$ vanishes, 
then 
\be
           \frac{ \pa x^\mu}{ \pa t} + (N^k)^\mu{}_\nu \, \frac{ \pa x^\nu}{ \pa y}  = 0 \qquad  \mu=1,\ldots,n \, , \quad k=1,2,\ldots
                   \label{N-hierarchy}
\ee
is known to be a hierarchy. 

The subject of the present work is a deformation of this hierarchy, presented by Lorenzoni and Magri in 
2005 \cite{Lore+Magri05} (also see \cite{Lore06,Lore+Perl23}). Here $N^k$ is replaced by the type $(1,1)$ tensor field
\bez
         M_k = N^k - \sum_{i=0}^{k-1} a_{k-1-i} N^i \, .
\eez 
This involves a chain of scalars $a_k$, $k=0,1,2,\ldots$, related by conservation laws of the new hierarchy.  
The Nijenhuis torsion of $M_k$ does not vanish, in general. Let us recall an important relevant result.

\begin{theorem}[Haantjes 1955 \cite{Haan55,Froe+Nije56}]
Let $M$ be a smooth type $(1,1)$ tensor field on a smooth manifold and such that there is a frame field 
(section of the bundle of linear frames) 
with respect to which the components of $M$ are diagonal. Then the existence of (local) coordinates, 
in which the components of $M$ are diagonal,  is equivalent to the vanishing of its Haantjes tensor $H(M)$. 
 \hfill $\Box$
\end{theorem}
 
The Haantjes tensor of a type $(1,1)$ tensor field $M$ vanishes if the Nijenhuis torsion $T(M)$ vanishes, but 
the vanishing of $H(M)$ is a weaker condition. Since we assumed $T(N)=0$, also $H(N) =0$. 
Under the assumptions of the theorem, which require that the eigenvalues of $N$ at each point are real, there 
are thus local coordinates in which the components of $N$ are diagonal. 
Clearly, $M_k$ is then also diagonal in these coordinates. As a consequence of Haantjes' theorem, 
also $H(M_k)=0$, $k=1,2,\ldots$. 

Haantjes' theorem says that a strictly hyperbolic hydrodynamic-type system, i.e., a system 
 with mutually distinct ``characteristic speeds"  \cite{Fera+Mars07}, is 
diagonalizable if and only if the corresponding Haantjes tensor vanishes. 
There is a generalization of Haantjes' theorem addressing block-diagonalization of a hydrodynamic-type system
\cite{Bogo06a}, where a weaker condition than vanishing of the Haantjes tensor is at work and it is assumed that 
the type $(1,1)$ tensor field of the hydrodynamic-type system maps each eigenspace into itself.  

Using available (block-) diagonalization theorems, apparently we cannot safely conclude that $H(M_k)=0$, without restrictions 
on $M_k$,  and it may therefore be worth to present a corresponding direct proof. In fact, it is also interesting to figure out, 
what the form of $T(M_k)$ is and in which way $H(M_k)=0$ is actually achieved. 

Section~\ref{sec:FN} collects necessary material from Fr\"olicher-Nijenhuis theory. Section~\ref{sec:hds} recalls 
some facts about hydrodynamic-type systems and a proof that (\ref{N-hierarchy}) is a hierarchy if $N$ has 
vanishing Nijenhuis tensor. In Section~\ref{sec:LM} we present the 
main content of \cite{Lore+Magri05} in a slightly different way, in the hope to somewhat more reveal the structure 
behind it. It includes an analysis of the structure of $T(M_k)$ and a direct proof of the vanishing of the 
Haantjes tensor $H(M_k)$.  Section~\ref{sec:conclusions} contains some final remarks.

\section{Some results from Fr\"olicher-Nijenhuis theory}
\label{sec:FN}

Let $N$ be a tensor field of type $(1,1)$ on an $n$-dimensional  smooth manifold $\cM$. In local coordinates (and using the
summation convention),
\bez
               N = N^\mu{}_\nu \, dx^\nu \otimes \frac{\partial}{\partial x^\mu} \, .
\eez
It determines a $C^\infty(\cM)$-linear map $\mathbb{X} \rightarrow \mathbb{X}$, where $\mathbb{X}$ is the space of vector fields on $\cM$ via
\bez 
           NX = N( X^\nu \frac{\partial}{\partial x^\nu} ) = X^\nu \, N(\frac{\partial}{\partial x^\nu} ) 
                                       = X^\nu \, N^\mu{}_\nu \, \frac{\partial}{\partial x^\mu} \, .
\eez
It also determines a $C^\infty(\cM)$-linear map $N^\ast : \bigwedge^1(\mathcal{M}) \rightarrow \bigwedge^1(\mathcal{M})$ 
of the space of differential 1-forms on $\cM$ via 
\bez
              N^\ast \alpha(X) := \alpha(N X) = (NX) \ins \alpha =: (N \ins \alpha)(X)  \, .
\eez
In local coordinates,
\bez
      N^\ast \alpha =  N^\ast ( \alpha_\mu \, dx^\mu ) = \alpha_\mu \, N^\ast dx^\mu  
                                                           = \alpha_\mu \, N^\mu{}_\nu \, dx^\nu \, .
\eez
The first appearance of $\ins$ in the preceding equation is the usual insertion of a vector field in a differential form (interior product), which is a graded derivation on the algebra $\bigwedge(\mathcal{M})$ of differential forms.   
$N \ins$ is defined as insertion of the vector field part of the type $(1,1)$ tensor $N$ in a differential form. This is a derivation:
\bez
        N \ins (\alpha \wedge \beta) = (N \ins \alpha) \wedge \beta + \alpha \wedge N \ins \beta 
\eez
for any differential forms $\alpha$ and $\beta$. 

If  $N_1,N_2$ are two tensor fields of type $(1,1)$, then 
\bez
            N_1 N_2 = N_1^\mu{}_\lambda \,   N_2^\lambda{}_\nu \, dx^\nu \otimes \frac{\partial}{\partial x^\mu} \, .
\eez
For 1-forms, we have $(N_1 N_2) \ins \alpha = N_2 \ins (N_1 \ins \alpha)$, i.e., 
$(N_1 N_2)^\ast \alpha =  N_2^\ast  N_1^\ast \alpha$.

Associated with any  type $(1,1)$ tensor field $N$ is a graded derivation $d_N$ of degree one \cite{Froe+Nije56}. 
For any differential form $\omega$, 
\bez
          d_N \omega = N \ins d \omega - d (N \ins \omega) \, ,
\eez
where $d$ is the exterior derivative. For a function $f$, we have $d_N f = N \ins df = N^\ast df$. 
If $N$ is the identity, then $d_N = d$. 
Furthermore, we have 
\bez
         d \, d_N = - d_N \, d   \, .
\eez

The \emph{Fr\"olicher-Nijenhuis bracket} of two type $(1,1)$ tensor fields is
\bez
    [N_1,N_2]_{\mathrm{FN}}(X,Y) &:=& (N_1 N_2 + N_2 N_1)[X,Y] + [N_1X,N_2Y] +[N_2X,N_1Y] \\
      & &   - N_1 ( [N_2X,Y] + [X,N_2Y] ) - N_2 ( [N_1X,Y] + [X,N_1Y] ) \, ,
\eez
where $X$ and $Y$ are any vector fields on $\mathcal{M}$. This bracket is bilinear and symmetric.

\begin{proposition}
\label{prop:[f1N1,f2N2]FN}
For type $(1,1)$ tensor fields $N_1,N_2$ and functions $f_1,f_2$,  
\bez
        [f_1 N_1, f_2 N_2]_{\mathrm{FN}} &=& f_1 f_2 \, [N_1, N_2]_{\mathrm{FN}} \\
   && + f_1 \Big( df_2(N_1 X) \, N_2 Y - df_2(X) \, N_1 N_2 Y - df_2(N_1 Y) \, N_2 X + df_2(Y) \, N_1 N_2 X \Big) \\
   && + f_2 \Big( df_1(N_2 X) \, N_1 Y - df_1(X) \, N_2 N_1 Y - df_1(N_2 Y) \, N_1 X + df_1(Y) \, N_2 N_1 X \Big) \, .
\eez
\end{proposition}
\begin{proof}
Direct computation.
\end{proof}

The \emph{Nijenhuis torsion} of a type $(1,1)$ tensor field $N$ is the type $(1,2)$ tensor field
\bez
     T(N)(X,Y) := \frac{1}{2} [N,N]_{\mathrm{FN}}(X,Y) = [NX,NY] + N \Big( N [X,Y] - [NX,Y] - [X,NY] \Big)  \, .
\eez
We have $d_N^2 =0$ if and only if $T(N)=0$ \cite{Froe+Nije56}.

\begin{proposition}
\label{prop:$[N^k,N^l]FN}
If $[N,N]_{\mathrm{FN}}=0$, then also $[N^k,N^l]_{\mathrm{FN}}=0$ for any positive integers $k,l$.
\end{proposition}
\begin{proof}
This follows immediately from Lemma 1.1 in \cite{Koba62}.
\end{proof}

In particular,  $T(N)=0$ implies $T(N^k)=0$ for any  positive integer $k$.

\begin{remark}
\label{rem:d_N1,d_N2_bi-diff}
For two different tensor fields $N_1,N_2$ of type $(1,1)$, satisfying $[N_i,N_j]_{\mathrm{FN}}=0$, $i,j=1,2$,   
the corresponding graded derivations of degree one, $d_{N_1}$ and $d_{N_2}$, satisfy 
\bez
        d_{N_1}^2 = 0 = d_{N_2}^2 \, , \qquad d_{N_1} \, d_{N_2} = - d_{N_2} \, d_{N_1} \, ,
\eez
and thus supply the algebra of differential forms with the structure 
of a \emph{bi-differential calculus} \cite{DMH00a}. 
\end{remark}

\subsection{Bi-closed 1-forms}
Throughout this section, $N$ is a type $(1,1)$ tensor field with $T(N)=0$. 

\begin{definition}
A 1-form $\rho$ is \emph{bi-closed} (with respect to a tensor field $N$ of type $(1,1)$) if
\bez
            d \rho = 0 \qquad \mbox{and} \qquad d_N \rho = 0 \, .   
           \vspace{-1.cm}
\eez
\end{definition}

A 1-form $\rho$ is bi-closed if and only if there are functions $f,h$ locally, such that
\bez
          \rho = df \, , \qquad     N^\ast d f = d h \, .
\eez 

\begin{remark}
A bi-closed 1-form (with respect to $N$) has also been called  a \emph{conservation law} for $N$. 
This has been substantiated in the framework of (finite-dimensional) bi-Hamiltonian systems. 
In Section~\ref{sec:hds} we will recall that a bi-closed 1-form 
is also a conservation law of a hydrodynamic-type system formed with $N$ 
\cite{Stone73,Grif+Mehd97}. 
Bi-closed 1-forms also appeared under the names ``fundamental 1-forms"  \cite{Magr+Moro84} and 
``Hamiltonian forms" (see, e.g., \cite{Bone15}).  
\end{remark}

\begin{proposition}[Stone 1973 \cite{Stone73}]
If a 1-form $\rho$ is bi-closed, then also $(N^\ast)^k \rho$, for any positive integer $k$. \hfill $\Box$
\end{proposition}

\subsubsection{A non-local chain of bi-closed 1-forms}
Let $T(N)=0$ and $\rho_1$ be a given bi-closed 1-form. Since $d \rho_1 = 0$ implies the existence 
of a function $a_1$ such that
\bez
   \rho_1 = d a_1 \, ,
\eez
it follows that 
\bez
    \rho_2 := d_N a_1 
\eez
is also a bi-closed 1-form:
\bez
       d \rho_2 = d d_N a_1 = - d_N d a_1 = - d_N \rho_1 = 0 \, , \quad d_N \rho_2 = d_N^2 a_1 = 0 \, .  
\eez
This leads to a sequence of bi-closed 1-forms,
\be
    \rho_{k+1} = d a_{k+1} = d_N a_k \, , \qquad k=1,2,\ldots \, .   \label{Lenard_chain}
\ee
To start with, we may choose a solution $a_0$ of the linear equation $d d_N a_0 =0$ and set $\rho_1 = d_N a_0$.

A sequence of the above form has sometimes been called a \emph{Lenard chain} \cite{Lax76,Kosm+Magr96,Magr03,Prau+Smir05}. 
Also see Theorem~3.4 in \cite{Gelf+Dorf79}. This works more generally for any bi-differential calculus \cite{DMH00a}. 
In particular, we may replace $d,d_N$ by $d_{N_1},d_{N_2}$ under the assumptions in Remark~\ref{rem:d_N1,d_N2_bi-diff}.

\begin{remark} 
We further note that this construction does \emph{not} make use of the graded derivation property of $d$ and $d_N$. 
Therefore it also works for any two anticommuting linear maps $D_1, D_2 : \mathbb{M}^r \rightarrow  \mathbb{M}^{r+1}$, 
acting on a graded linear space $\mathbb{M}$, 
provided they are flat, i.e., $D_1^2 = 0 = D_2^2$, and that $D_1 \rho = 0$  for (at least some) $\rho \in  \mathbb{M}^1$ implies 
that there is a $\chi \in  \mathbb{M}^0$ such that $\rho = D_1 \chi$. Also see \cite{DMH00e}.
\end{remark}

\section{Some fact about systems of hydrodynamic type}
\label{sec:hds}
Associated with any type $(1,1)$ tensor field $M$ on an $n$-dimensional smooth manifold $\cM$ is an autonomous 
system of partial differential equations,
\be
       \frac{\pa x^\mu}{\pa t} + M^\mu{}_\nu \, \frac{\pa x^\nu}{\pa y}  = 0 \qquad \mu =1,\ldots,n
                  \, ,  \label{hydrodyn_sys}
\ee
which is of ``hydrodynamic type" (see, e.g., \cite{Tsar91}). Here $t$ and $y$ are two independent variables. 

If there are functions $f$ and $h$ such that 
\be
                  M^\ast df = dh \, ,    \label{hd_conserv_law}
\ee
thus
\bez
                   \frac{\pa h}{\pa x^\nu} = \frac{\pa f}{\pa x^\mu} \, M^\mu{}_\nu  \, ,
\eez
then we obtain
\bez
         \frac{\pa f}{\pa t} + \frac{\pa h}{\pa y} = \frac{\pa f}{\pa x^\mu} \,  \frac{\pa x^\mu}{\pa t}
     + \frac{\pa h}{\pa x^\nu} \, \frac{\pa x^\nu}{\pa y} 
  = \frac{\pa f}{\pa x^\mu} \, \left( \frac{\pa x^\mu}{\pa t} + M^\mu{}_\nu \, \frac{\pa x^\nu}{\pa y}  \right) = 0 
\eez
along any solution $x^\mu(t,y)$ of the above hydrodynamic-type system. The last equation has the 
familiar form of a conservation law. Hence, any pair of functions satisfying  (\ref{hd_conserv_law}) is a conservation 
law for the above hydrodynamic-type system. 

If $T(M)=0$, given a bi-closed 1-form $\rho$, it implies the local existence of functions $f,h$ such that
 (\ref{hd_conserv_law}) holds. But the latter does not require $T(M)=0$.

\subsection{Compatible systems of hydrodynamic type}

Let us assume that two systems of hydrodynamic type,
\bez
           \frac{\pa x^\mu}{\pa s} + A^\mu{}_\nu \, \frac{\pa x^\nu}{\pa y}  = 0 \, , \qquad
           \frac{\pa x^\mu}{\pa t} + B^\mu{}_\nu \, \frac{\pa x^\nu}{\pa y}  = 0 \, ,
\eez
admit a common solution $x^\mu(s,t,y)$. Then
\bez
      0 &=& \frac{\pa^2 x^\mu}{\pa s \pa t}  -  \frac{\pa^2 x^\mu}{\pa t \pa s} \\
    &=& (  A^\mu{}_{\kappa,\nu} \, B^\nu{}_\lambda + A^\mu{}_\nu \, B^\nu{}_{\kappa,\lambda}
            - B^\mu{}_{\kappa,\nu} \, A^\nu{}_\lambda - B^\mu{}_\nu \, A^\nu{}_{\kappa,\lambda} )
             \,  \frac{\pa x^\kappa}{\pa y} \,  \frac{\pa x^\lambda}{\pa y} \\
    && +(A^\mu{}_\nu \, B^\nu{}_\kappa - B^\mu{}_\nu \, A^\nu{}_\kappa) \, \frac{\pa^2 x^\kappa}{\pa y^2}
\eez
along the solution.
If there is a sufficient number of common solutions, this implies that 
$A$ and $B$ have to commute, as matrices, and they have to satisfy
\bez
   B^\nu{}_{(\kappa} \, A^\mu{}_{\lambda),\nu} \,  + A^\mu{}_\nu \, B^\nu{}_{(\kappa,\lambda)}
            - A^\nu{}_{(\kappa} B^\mu{}_{\lambda),\nu} \,  - B^\mu{}_\nu \, A^\nu{}_{(\kappa,\lambda)} = 0 \, .
\eez
Here a comma indicates a partial derivative and round brackets symmetrization. 
Contracting the last expression with the components $X^\kappa, X^\lambda$ of an arbitrary vector field $X$, 
and using $[A,B]=0$ (commutator of matrices), it can be cast into the form
\bez
                 [A,B]_X = 0   \qquad \forall X \in \mathbb{X} \, ,
\eez
where
\bez
                 [A,B]_X := [AX,BX] - A \,[X,BX] + B \, [X,AX]  \, . 
\eez
Also see \cite{PSS96}. 
The bracket defined in this way (with respect to a vector field $X$) is bilinear and antisymmetric. 
If $I$ denotes the identity, then $[A,I]_X =0$. For a function $f$, we have
\bez
     [A, f B]_X &=& f \, [A,B]_X + [AX, f] \, BX - [X,f] \, ABX \\
                     &=& f \, [A,B]_X + (AX f) \, BX - (Xf) \, ABX \\
       &=& f \, [A,B]_X + df(AX) \, BX - df(X) \, ABX  \, .
\eez
In particular,
\bez
      [f,g]_X = 0 \, ,
\eez
for functions $f,g$ (regarded as multiplication operators).

\begin{definition}
Two systems of hydrodynamic type with $(1,1)$ tensor fields $A$ and $B$ on a manifold $\cM$ are 
called \emph{compatible} if $A$ and $B$ commute as matrices and if $[A,B]_X =0$  
for all vector fields $X$ on $\cM$. 
\end{definition}

\subsection{A hydrodynamic-type hierarchy}

We recall a well-known result.

\begin{proposition}
\label{prop:N^k_hierarchy}
Let $N$ be a tensor field of type $(1,1)$ with $T(N)=0$.
The family of hydrodynamic-type flows
\bez
      \frac{\pa x^\mu}{\pa t_k} + (N^k)^\mu{}_\nu \, \frac{\pa x^\nu}{\pa y}  = 0 \qquad k=1,2,\ldots 
\eez 
forms a \emph{hierarchy}, i.e., any two members are compatible.
\end{proposition}
\begin{proof} Clearly, the matrices formed by the components of powers of $N$ commute. Furthermore,
\bez
  [N^{i+1},N^{j+1}]_X &=&  [N N^i X, N N^j X] - N^{i+1} \, [X,N^{j+1} X] + N^{j+1} \, [X,N^{i+1} X] \\
    &=& - N^2 [N^i X,N^j X] + N \big( [N^{i+1} X , N^j X] + [N^i X, N^{j+1} X]   \\
     &&    - N^i \, [X,N^{j+1} X] + N^j \, [X,N^{i+1} X] \big) \qquad \mbox{using $T(N)(N^i X,N^j X)=0$} \\
    &=& -N^2 \big( [N^i,N^j]_X + N^i [X,N^j X] - N^j [X,N^i X] \big) \\
    &&   + N \big( [N^{i+1},N^j]_X + [N^i, N^{j+1}]_X
            + N^{i+1} [X,N^j X] - N^{j+1} [X,N^i X] \big) \\
    &=&  -N^2 [N^i,N^j]_X + N \big( [N^{i+1},N^j]_X + [N^i, N^{j+1}]_X \big)  \, .
\eez
 From this recursion we obtain $[N^i, N^j]_X =0$, for $i,j=0,1,2,\ldots$, by induction. 
Suppose that the brackets up to $[N^i,N^j]_X$ vanish. Then the above recursion yields
\bez
       [N^{i+1},N^j]_X = N \, [N^{i+1},N^{j-1}]_X \, ,
\eez
which, by iteration, leads to
\bez
       [N^{i+1},N^j]_X = N^j \, [N^{i+1},I]_X = 0 \, .
\eez
Now the above recursion reduces to 
\bez
      [N^{i+1},N^{j+1}]_X = N \, [N^i, N^{j+1}]_X = \ldots = N^{i+1} [I,N^{j+1}]_X = 0 \, .
\eez
\end{proof}

\section{The Lorenzoni-Magri hierarchy}
\label{sec:LM}

\begin{lemma}
\label{lem:M_recursion}
Let $T(N)=0$ and $A_k$, $k=0,1,2,\ldots$, tensor fields of type $(1,1)$. Let
\bez
       M_{k+1} = N M_k + A_k \, , \qquad M_0 = I \, .
\eez
Then
\bez
      [M_{i+1} , M_{j+1}]_X &=& - N^2 [M_i , M_j]_X + N \left( [M_{i+1}, M_j]_X - [M_{j+1} , M_i ]_X 
             - [A_i,M_j]_X + [A_j,M_i]_X \right) \\
      &&  + [M_{i+1} , A_j]_X - [M_{j+1} , A_i ]_X  - [A_i,A_j]_X \, .
\eez
\end{lemma}
\begin{proof} Using bilinearity of the bracket,
\bez
     [M_{i+1} , M_{j+1}]_X = [N M_i , N M_j]_X + [N M_i,A_j]_X + [A_i , NM_j]_X + [A_i,A_j]_X \, .
\eez   
Expanding the first term on the right hand side according to its definition and then using vanishing 
Nijenhuis torsion in the form 
$T(N)(M_i X, M_j X)=0$, we obtain
\bez
    [M_{i+1} , M_{j+1}]_X &=& N \, \big( [M_iX,NM_jX] - [M_j X,N M_i X] - M_i [X,NM_jX] + M_j [X,NM_iX]  \\
   &&  - N \, [M_iX,M_jX]  \big) + [N M_i,A_j]_X + [A_i , NM_j]_X + [A_i,A_j]_X  \, .
\eez
Reformulating the terms in the round brackets back in terms of brackets $[ \, , \, ]_X$, we  
reach the stated formula. We note that the steps are essentially the same as those in the proof of 
Proposition~\ref{prop:N^k_hierarchy}. 
\end{proof}

If $T(N)=0$, then the hydrodynamic-type systems associated with any two polynomials in $N$ with constant 
coefficients are compatible. 
The following is a generalization of Proposition~\ref{prop:N^k_hierarchy} and shows that the coefficients 
need not be constant in order to have a hierarchy.

\begin{theorem}[Lorenzoni and Magri 2005 \cite{Lore+Magri05}]
\label{thm:Lorenzoni-Magri}
Let $T(N)=0$ and
\bez
     M_{k+1} = N M_k - a_k \, I \, , \qquad
     d a_k = M_k^\ast \, d a_0 \, ,
\eez
with $M_0 = I$ and scalars $a_k$, $k=0,1,2,\ldots$. Then
\bez
      \frac{\pa x^\mu}{\pa t_k} + (M_k)^\mu{}_\nu \, \frac{\pa x^\nu}{\pa y}  = 0 \qquad k=0,1,2,\ldots 
\eez 
constitute a hierarchy.
\end{theorem}
\begin{proof}
Setting $A_k = - a_k \, I$ in Lemma~\ref{lem:M_recursion}, we have
\bez
   &&  [M_{i+1} , M_{j+1}]_X +N^2 [M_i,M_j]_X - N \, (  [M_{i+1},M_j]_X - [M_{j+1},M_i ]_X )  \\
  &=& N ( [M_i,a_j]_X - [M_j,a_i]_X) - [M_{i+1} , a_j] + [M_{j+1},a_i]_X \\
  &=&  ( - da_i(M_j X) + da_j(M_i X) ) \, NX - da_i(X) \, (M_{j+1} - NM _j) X \\
     &&      + da_j(X) \, (M_{i+1} - NM_i) X  
           + ( da_i(M_{j+1} X) - da_j(M_{i+1} X) ) \, X  \\
  &=& ( da_i(M_{j+1} X) - da_j(M_{i+1} X) + a_i \, da_j(X) - a_j \, da_i(X) ) \, X \\
     &&   - ( da_i(M_j X) - da_j(M_i X) ) \, NX   \\
   &=&  ( da_i(M_j N X) - da_j(M_i N X) ) \, X - ( da_i(M_j X) - da_j(M_i X) ) \, NX   \\
   &=&  ( M_j^\ast da_i(N X) - M_i^\ast da_j(N X) ) \, X - ( M_j^\ast da_i(X) - M_i^\ast da_j(X) ) \, NX       \, .
\eez
We observe that the last expression vanishes if we set $d a_k = M_k^\ast d a_0$, since then
\bez
      M_i^\ast d a_j = M_i^\ast \, M_j^\ast \, da_0 = M_j^\ast \, M_i^\ast \, da_0 
                                  = M_j^\ast da_i    \qquad i,j=0,1,2, \ldots \, .
\eez
Using the resulting recursion relation (which has the same form as the one in the proof of
Proposition~\ref{prop:N^k_hierarchy}), $[M_i,M_j]_X=0$ follows by induction.
\end{proof}

For explorations of hydrodynamic-type hierarchies of the above form, see \cite{Lore+Magri05,Lore06,Lore+Perl23}.

\begin{remark}
We note that, for $k=1,2,\ldots$, the relation $d a_k = M_k^\ast \, d a_0$ in the last proposition is
a conservation law of the hydrodynamic-type system with $M_k$. 
Furthermore, it implies 
\bez
         d a_{k+1} = M_{k+1}^\ast da_0 = N^\ast M^\ast_k da_0 - a_k \, da_0 = N^\ast da_k - a_k \, da_0 
                            = d_N a_k  - a_k \, da_0 \, ,
\eez 
which is (\ref{Lenard_chain}) modified by the last term.
Regarding $da_0$ as the gauge potential 1-form of a connection, we can introduce the 
covariant  exterior derivative
\bez
            D_N := d_N - da_0 \wedge \, ,
\eez
acting on $\bigwedge(\cM)$. $D_N$ is $\mathbb{R}$-linear and flat, i.e., $D_N^2 = 0$. It also satisfies $d \, D_N = - D_N \, d$. 
This makes contact with the approach to 
conservation laws of integrable systems in the bi-differential calculus, respectively bi-complex framework 
\cite{DMH00a,DMH00e}. 
The above recursion relation can now be more compactly expressed as 
\be
                d a_{k+1} = D_N a_k  \, .    \label{module_conserv_law_chain}
\ee
\end{remark}

\begin{remark}
The $M_k$ determine graded derivations of degree one, $d_{M_k}$, but with $d_{M_k}^2 \neq 0$.
For example, using the identities
\bez
        d_{N_1 + N_2} = d_{N_1} + d_{N_2} \, , \qquad 
        d_{f N} = f \, d_N - df \wedge N \ins  \, ,
\eez
on $\bigwedge(\mathcal{M})$,  for tensor fields of type $(1,1)$ and a function $f$, we find 
\bez
      d_{M_1} = d_N - d_{a_0 I} = d_N - a_0 \, d + da_0 \wedge I \ins \, ,
\eez
which satisfies
\bez
      d_{M_1}^2 \omega = d_N^2 \, \omega + r \, d_N da_0 \wedge \omega + da_0 \wedge d_N \omega - d_N a_0 \wedge d \omega 
\eez
for any $r$-form $\omega$. We find that
\bez
           d_{M_1} a_0 = D_N a_0 \, ,
\eez
in accordance with $d_{M_k} a_0 = M_k^\ast \, da_0 = d a_k$ and (\ref{module_conserv_law_chain}). 
\end{remark}

\subsection{The Haantjes tensor}

The Nijenhuis torsion of $M_k$ (with $N \neq I$) does \emph{not} vanish unless the functions 
$a_k$ are constants. 
In particular, using Proposition~\ref{prop:[f1N1,f2N2]FN} and bilinearity of the Fr\"olicher-Nijenhuis bracket, we obtain
\bez
    T(M_1)(X,Y) = - [N,a_0]_{\mathrm{FN}}(X,Y) 
                           = - da_0(NY) \, X + da_0(NX) \, Y + da_0(Y) \, NX  - da_0(X) \, NY   \, .
\eez
The \emph{Haantjes tensor} of a type $(1,1)$ tensor field $M$ is given in terms of its Nijenhuis torsion by
\bez
           H(M)(X,Y) := M^2 T(M)(X,Y) + T(M)(MX,MY) - M \big( T(M)(MX,Y) + T(M)(X,MY) \big) \, .
\eez
Using the general property
\bez
           H(M + f \, I) = H(M) 
\eez
of the Haantjes tensor (Proposition~1 in \cite{Bogo96}), for any type $(1,1)$ tensor field $M$ and 
any function $f$, we have
\bez
           H(M_1) = H(N) = 0 \, .   
\eez

\begin{theorem}
Let type $(1,1)$ tensor fields $M_k$, $k=1,2,\ldots$, be given by the recursion relation in 
Theorem~\ref{thm:Lorenzoni-Magri}, with functions $a_k$ and $T(N)=0$.\footnote{The recursion relation in 
Theorem~\ref{thm:Lorenzoni-Magri} for the functions $a_k$ is not used.}
Then $H(M_k) =0$,  $k=1,2,\ldots$.  
  \hfill $\Box$
\end{theorem}

This follows as a consequence of the next two results, the first giving the crucial insight.

\begin{lemma}
Let $M$ and $N$ be two commuting tensor fields of type $(1,1)$. If the Nijenhuis torsion of $M$ has the form
\be
       T(M)(X,Y) = \sum_{i=0}^m \Big( f_i(x,N,Y) \, N^i X - f_i(x,N,X) \, N^i Y \Big) \, ,    \label{T_with_H=0}
\ee
with some non-negative integer $m$ and functions $f_i$, then the Haantjes tensor of $M$ vanishes.
\end{lemma}
\begin{proof}
This is easily verified directly.
\end{proof}

\begin{proposition}
$T(M_k)$ has the form (\ref{T_with_H=0}).
\end{proposition}
\begin{proof}
The recursion relation for the $M_k$ implies that
\bez
         M_{k+1} = N^{k+1} - \sum_{i=0}^k a_{k-i} N^i \, .
\eez 
Then 
\bez
       T(M_k) &=& \frac{1}{2} \, [ M_k , M_k ]_{\mathrm{FN}} \\
     &=& \frac{1}{2} \, [ N^k , N^k ]_{\mathrm{FN}} - \sum_{i=0}^{k-1} [ N^k , a_{k-1-i} N^i  ]_{\mathrm{FN}} 
             + \frac{1}{2} \, \sum_{i=0}^{k-1} \sum_{j=0}^{k-1} [ a_{k-1-i} N^i  , a_{k-1-j} N^j  ]_{\mathrm{FN}} \, .
\eez
Since $T(N)=0$, our assertion now follows from Propositions~\ref{prop:[f1N1,f2N2]FN} and ~\ref{prop:$[N^k,N^l]FN}.
\end{proof}

\section{Conclusions}
\label{sec:conclusions}

The Lorenzoni-Magri hierarchy is an example that shows that Nijenhuis geometry, 
dealing with $(1,1)$ tensor fields with vanishing Nijenhuis torsion (also see \cite{BKM22,BKM23}), is a too narrow 
framework  for (integrable) hydrodynamic-type systems. Rather, it should be 
widened at least to Haantjes geometry, where the weaker condition of vanishing Haantjes tensor is imposed. This 
is not a new insight. It has led to the notion of Haantjes manifolds \cite{Magr18} and other developments,
see, in particular, \cite{Temp+Tond22a,Temp+Tond22b,RTT23}. 
In \cite{Bogo06b} it was shown that vanishing of the Haantjes tensor is a necessary condition for the existence of a 
bi-Hamiltonian structure of a hydrodynamic-type system. 
\vspace{.5cm}

\noindent
\textbf{Important note added.} An excellent referee pointed out that Theorem~4.5, considered to be a main 
result of this work, follows directly from Corollary 3.3 in O. I. Bogoyavlenskij, General algebraic identities for 
the Nijenhuis and Haantjes tensors, Izvestiya: Mathematics, 2004, Volume 68, 1129--1141.

\end{document}